\documentclass[letterpaper, 10 pt, conference]{ieeeconf}  

\IEEEoverridecommandlockouts                              

\usepackage{cite}
\usepackage{amsmath,amssymb,amsfonts}
\usepackage{algorithmic}
\usepackage{graphicx}
\usepackage{textcomp}
\usepackage{xcolor}
\usepackage{enumerate}
\usepackage[mathscr]{euscript}

\usepackage[utf8]{inputenc}
\usepackage[english]{babel}

\DeclareMathAlphabet{\mathpzc}{OT1}{pzc}{m}{it}

\DeclareMathOperator*{\argmin}{arg\,min}

\DeclareMathOperator*{\arginf}{arg\,inf}
\newcommand{\bmin}{\ensuremath{S_{\rm L}}}
\newcommand{\bmax}{\ensuremath{S_{\rm U}}}
\newcommand{\Lnash}{\ensuremath{\mathcal{L}^{\rm nf}}}
\newcommand{\Lopt}{\ensuremath{\mathcal{L}^{\rm opt}}}
\newcommand{\fnash}{\ensuremath{f^{\rm nf}}}
\newcommand{\fopt}{\ensuremath{f^{\rm opt}}}
\newcommand{\poa}{\ensuremath{{\rm PoA}}}
\newcommand{\gee}{\mathcal{G}}
\newcommand{\qed}{\hfill\blacksquare}

\newcommand{\geelc}{\mathcal{G}^{\rm lc}}
\newcommand{\sdist}{\mathscr{S}}
\newcommand{\sdistm}{\mathscr{S}(\sbar)}
\newcommand{\sdistmbi}{\mathscr{S}^{\rm bi}(\sbar)}
\newcommand{\sbar}{\overline{s}}
\newcommand{\slow}{s_\mathpzc{l}^{(\sbar,G,k)}}
\newcommand{\supp}{s_\mathpzc{u}^{(\sbar,G,k)}}
\newcommand{\Tsmc}{T}
\newcommand{\kagn}{k_{\rm agn}}
\newcommand{\kg}{k_{(G)}}
\newcommand{\ksg}{k_{(\sbar,G)}}
\newcommand{\ks}{k_{(\sbar)}}
\newcommand{\slone}{S_{\mathpzc{l}1}}

\newcommand{\sutwo}{S_{\mathpzc{u}2}}
\newcommand{\Tau}{\mathcal{T}}

\newtheorem{theorem}{Theorem}
\newtheorem{lemma}{Lemma}[theorem]
\newtheorem{corollary}{Corollary}

\newtheorem{proposition}{Proposition}

\title{Utilizing Information Optimally to Influence Distributed Network Routing}
\author{{Bryce L. Ferguson, Philip N. Brown, and Jason R. Marden}
\thanks{This research was supported by ONR grant \#N00014-17-1-2060 and NSF grant \#ECCS-1638214}
\thanks{B. L. Ferguson (corresponding author) and J. R. Marden are with the Department of Electrical and Computer Engineering, University of California, Santa Barbara, CA, {\texttt{\{blferguson,jrmarden\}@ece.ucsb.edu}}.}
\thanks{P. N. Brown is with the Department of Computer Science, University of Colorado at Colorado Springs, {\texttt{philip.brown@uccs.edu}}}  
}

\begin{document}
\maketitle

\begin{abstract}
How can a system designer exploit system-level knowledge to derive incentives to optimally influence social behavior?
The literature on network routing contains many results studying the application of monetary tolls to influence behavior and improve the  efficiency of self-interested network traffic routing.
These results typically fall into two categories: (1) optimal tolls which incentivize socially-optimal behavior for a known realization of the network and population, or (2) robust tolls which provably reduce congestion given uncertainty regarding networks and user types, but may fail to optimize routing in general.
This paper advances the study of robust influencing, mechanisms asking how a system designer can optimally exploit additional information regarding the network structure and user price sensitivities to design pricing mechanisms which influence behavior.
We design optimal scaled marginal-cost pricing mechanisms for a class of parallel-network routing games and derive the tight performance guarantees when the network structure and/or the average user price-sensitivity is known.
Our results demonstrate that from the standpoint of the system operator, in general it is more important to know the structure of the network than it is to know distributional information regarding the user population.
\end{abstract}

\section{Introduction}
In systems whose performance is driven by social behavior, it is well known that the self-interested choices of individuals can dramatically degrade system performance. This inefficiency resulting from selfish behavior is commonly characterized by the ratio between the worst-case social welfare resulting from choices of self-interested users and the optimal social welfare; this is typically referred to as the \emph{price of anarchy}~\cite{Papadimitriou2001} and has become a highly studied area in resource allocation~\cite{Johari2004,Paccagnan2018a}, distributed control~\cite{Marden2014}, and transportation~\cite{Piliouras2013,Moradipari2018}. 
A common line of research studies how this inefficiency can be mitigated by using pricing mechanisms and information systems which incentivize users to make decisions more in line with the social optimum~\cite{Ratliff2018, Acemoglu2018}.
Naturally, an effective implementation of incentives is heavily dependent on how the users respond to the incentives.

In this paper, we study the design of incentive mechanisms to influence social behavior in a class of congestion games; our particular focus here is on the relationship between a designer's ability to effectively influence selfish behavior and the designer's uncertainty regarding various system parameters.

Specifically, we consider a network routing problem in which a unit mass of users need to be routed across a network of two parallel edges with affine latency functions. Finding a flow that minimizes the total latency in the system is straightforward if the system designer has full control over deciding the path of each user. However, a network flow resulting from users' self-interested decisions need not be optimal~\cite{Wardrop1952,Pigou1920,Roughgarden2002,Cole2018}.
Modeling the routing problem as a non-atomic congestion game, we characterize this self-interested behavior as a \emph{Nash flow}; that is, a network flow in which every user individually chooses the lowest-latency path, given the choices of other users.

Financial incentives in the form of road tolls are a common approach to mitigate the inefficiency due to selfish behavior; these tolls modify users' preferences, inducing new, more-efficient Nash flows.
In the case where the system designer is fully aware of the network topology, population size, and the users' price-sensitivity, tolls can be designed that influence users to self-route in a manner that minimizes total latency~\cite{Cole2003,Fleischer2004,Karakostas2004}. Though these tolls can induce optimal behavior, the required information can be difficult or impossible to obtain, and lack robustness to uncertainty in various system parameters~\cite{Brown2017a}.

Alternatively, in the case where the system designer is unaware of the structure of the network or population size, and knows only the possible support of user price-sensitivity, there exist robust taxation mechanisms which improve the efficiency for certain classes of networks. These tolls are not guaranteed to induce perfectly optimal behavior, but they bound the price of anarchy strictly below the nominal un-influenced value of $4/3$~\cite{Brown2017d}. The study of information value in system design has been considered in areas such as optimal control~\cite{Farokhi2015} and multi agent systems~\cite{Paccagnan2018b}.


In this paper, we seek to bridge the gap between optimal taxation mechanisms which require detailed information, and robust tolls that require less information but fail to perfectly optimize routing. 
In~\cite{Brown2017d}, the system designer is oblivious to the network structure and knows only the possible support of the distribution of user price-sensitivities; using only this information, the designer selects optimal robust tolls to minimize the price of anarchy.
Our contributions augment this information in three ways: one in which the designer additionally knows the network structure,  one in which the designer additionally knows the average user-sensitivity, and one in which both structure and average sensitivity are known.
For each information environment, we determine the tolling scheme that uses all available information \emph{optimally}, and the resulting efficiency guarantees are compared to previously-known results for robust tolls.

In comparing the performance of these mechanisms, we ask which piece of information is more valuable to the system designer when used optimally: network structure or users' average price sensitivity.
Intriguingly, the answer is highly context-dependent.
In some cases, we find that the mean-sensitivity is highly informative and is thus a more valuable piece of information than the network structure.
However, taken in worst-case, we find that the network structure is more valuable than the mean.

\section{Model and Performance Metrics}

\subsection{Routing Game}
Consider a routing problem, in which a unit mass of traffic must be directed through a parallel network from an origin node to a destination node over a set of edges $E$. A \textit{feasible flow} over the network is an assignment of traffic to each edge $f = \{f_e\}_{e\in E} \in \Delta (E)$ where $f_e \geq 0$ denotes the flow on an edge $e$ and $\Delta (E)$ denotes the standard probability simplex over the set $E$; that is, $\sum_{e\in E} f_e = 1$. To characterize transit delay, each edge $e \in E$ in the network has a latency function of the form
\begin{equation}
\ell_e(f_e) = a_e f_e + b_e.
\end{equation}
where $a_e \geq 0$ and $b_e \geq 0$. The latency on an edge is thus a non-decreasing, non-negative function of the flow on that edge. The system cost of a flow $f$ is characterized by the \textit{total latency} in the network, defined as
\begin{equation} \label{eq_total_lat}
\mathcal{L}(f) = \sum_{e \in E} f_e \cdot \ell_e(f_e),
\end{equation}
and we denote the flow that minimizes this total latency as $\fopt \in \argmin_{f \in \Delta E} \mathcal{L}(f)$. We specify a particular network by the tuple $G=(E,\{\ell_e\}_{e\in E})$.

This paper studies taxation mechanisms designed to influence the emergent collective behavior of self-interested, price-sensitive users. We model this routing problem as a congestion game where each edge $e \in E$ is assigned a flow dependent tolling function $\tau_e:[0,1] \rightarrow \mathbb{R}^+$. A user $x \in [0,1]$ has a price-sensitivity $s(x) >0$; this price-sensitivity is subjective for each user and relates the user's cost from being tolled to their cost from experiencing delays and is the reciprocal of the user's value of time. We write $e(x)$ to denote the edge chosen by user $x$. The cost function for a user $x$ that is on an edge $e(x) \in E$ can be expressed as
\begin{equation}
J_{x}(f) =  \ell_{e(x)}\left(f_{e(x)}\right)+s(x) \tau_{e(x)}\left(f_{e(x)}\right).
\end{equation}
A flow $f$ is a \textit{Nash flow} if for every user $x \in [0,1]$
\begin{equation}
J_{x}(f) \in \argmin_{e \in E} \{ \ell_e\left(f_{e}\right)+s(x) \tau_e\left(f_{e}\right) \}.
\end{equation}

A game is therefore characterized by a network $G$, player sensitivity distribution $s:[0,1] \rightarrow \mathbb{R}^+$, and a set of tolling functions $\{\tau_e\}_{e\in E}$, denoted by the tuple $(G,s,\{\tau_e\}_{e\in E})$. It is shown in \cite{Mas-Colell1984} that a Nash flow will always exist in a congestion game of this form.

\begin{figure*}
\centering
\vspace{3mm}
\begin{tabular}{| c || c | c | c | c ||}
\hline 
 & \begin{tabular}{@{}c@{}} \emph{Any sensitivity} \\ \emph{distribution} \end{tabular} & \begin{tabular}{@{}c@{}} $\sdist$ \\ \emph{mean-agnostic} \end{tabular} & \begin{tabular}{@{}c@{}} $\sdistm$ \\ \emph{mean-aware} \end{tabular} & \begin{tabular}{@{}c@{}} $s(x)$ \\ \emph{population-aware} \end{tabular} \\
\hline
\begin{tabular}{@{}c@{}} $\gee$ \\ \emph{network-agnostic}\vspace{0.1 in} \end{tabular} &
 \begin{tabular}{@{}c@{}}$\poa \left( \gee \right) =1.\overline{33}$ \\ \multicolumn{1}{r}{\cite{Roughgarden2002}} \end{tabular} &
 \begin{tabular}{@{}c@{}}$\poa^*  \left( \gee, \sdist \right) \approx 1.176$ \\ \multicolumn{1}{r}{(A),\cite{Brown2017d}} \end{tabular} &
 \begin{tabular}{@{}c@{}}$\poa^*  \left( \gee, \sdistm \right) \approx 1.1385$ \\ \multicolumn{1}{r}{(B)} \end{tabular}  &
   \\
\hline
\begin{tabular}{@{}c@{}}$G$ \\ \emph{network-aware}\vspace{0.1 in} \end{tabular} &
 & 
 \begin{tabular}{@{}c@{}}$\poa^*  \left( G, \sdist \right) \leq 1.09$ \\ \multicolumn{1}{r}{(C)} \end{tabular} & 
 \begin{tabular}{@{}c@{}}$\poa^*  \left( G, \sdistm \right) \leq 1.0494$ \\ \multicolumn{1}{r}{(D)} \end{tabular} &  
 \begin{tabular}{@{}c@{}}$\poa \left( G,T \right) =1$ \\ \multicolumn{1}{r}{\cite{Fleischer2004}} \end{tabular} \\
\hline
\end{tabular}
\caption{Price of anarchy bounds over different classes of games with optimal taxation mechanisms. Each class of games consists of parallel networks with affine latency functions. Information available to the system designer is embedded in the class of games (e.g. the network-aware classes only contain games with the same network). Smaller classes denote more certainty in the structure of the composing games, allowing for more effective toll design. These numbers are generated for $\bmin = 1$ and $\bmax = 10$, and a worst case over networks and mean sensitivities.}\vspace{-3mm}
\label{fig_table}
\end{figure*}

\subsection{Taxation Mechanisms \& Performance Metrics}
To understand the robustness of a tolling scheme, we consider the performance over a class of networks and users' sensitivities. For a network, $G$, we identify the latency functions which constitute the network by $L(G)$; further, for a family of networks $\gee$, let $L(\gee ) = \bigcup_{G\in \gee} L(G)$ be the set of all latency functions that compose the networks in $\gee$.

A \emph{taxation mechanism} $T$ maps latency functions $\ell_e$ to tolling functions $\tau_e$. For a family of networks $\gee$, this mapping is denoted  $T:L(\gee) \rightarrow \Tau$, where $\Tau$ is the set of all linear tolling functions on $[0,1]$. These tolling functions are termed \emph{scaled marginal-cost tolls} and as discussed in Section~\ref{section_main_results}, are an important class of tolling functions which are optimal in many contexts. For a class of games $\gee$, a taxation mechanism assigns a tolling function to an edge by the latency function on said edge, independently of what specific network in $\gee$ is realized. It is important to note that our formulation fits that of the classic Pigouvian-taxes~\cite{Pigou1920} given by
\begin{equation}
    \tau_e(f) = f_e \cdot \frac{d}{df_e}\ell_e(f_e) = a_e f_e,   \quad \forall \ e \in E.
\end{equation}
To formalize the notion of uncertainty in users' response to a taxation mechanism, we consider families of sensitivity distributions that can occur when the system designer is only aware of the lower bound $\bmin$ and upper bound $\bmax$ on users price sensitives. We define the set of possible sensitivity distributions as $\sdist = \{s:[0,1] \rightarrow [\bmin, \bmax] \}$. When the average price sensitivity $\sbar$ of the users is introduced to the system designer, the set of possible distributions becomes $\sdistm = \{s \in \sdist | \int_0^1 s(x)dx = \sbar \}$; it is clear that $\sdistm \subset \sdist$.

To evaluate the performance of a tolling mechanism, let $\Lnash(G,s,T)$ be the total latency on a network $G$, with price sensitivity distribution $s$, in the Nash flow $\fnash$ when tolls are assigned according to taxation mechanism\footnote{The taxation mechanism is a mapping from latency functions to tolling functions. A game with taxation mechanism $T$ is therefore denoted $(G,s,\{T(\ell_e)|\ell_e \in G\})$. For brevity, we simply denote this as $(G,s,T)$.} $T$, and let $\Lopt(G)$ be the minimum total latency which occurs under the optimal flow $\fopt$. The \textit{price of anarchy} compares the Nash flow on a network with the optimal flow; this characterizes the inefficiency of the network and can be defined as
\vspace{-2mm}
\begin{equation}
	\poa(G,s,T) = \frac{\Lnash(G,s,T)}{\Lopt(G)} \geq 1.
\end{equation}
We extend this definition to include families of networks and sensitivity distributions, i.e.,
\begin{equation}\label{eq_poa_def}
        \poa (\gee,\sdist,T)= \underset{G\in \gee}{\rm sup} \  \underset{s\in \sdist}{\rm sup}  \Bigg\{ \frac{\Lnash(G,s,T)}{\Lopt(G)} \Bigg\},
\end{equation}
such that the price of anarchy is now the worst-case inefficiency over possible networks and sensitivity distributions\footnote{A trivial taxation mechanism of assigning zero tolling functions to each edge (i.e., $\tau_e =  0$ for all $e\in E$) will have a price of anarchy of 4/3 for this class of networks~\cite{Roughgarden2002}.}. Note that the same taxation mechanism $T$ is applied to any realized network-sensitivity distribution pair.

\subsection{Optimal Tolling \& Our Contributions}

When designing a taxation mechanism, the goal of the system designer is to minimize worst-case inefficiency while the network and/or users' sensitivities are not known. Thus, we define an optimal tolling mechanism as
\begin{equation*}
T^* \in \underset{T: L(\gee) \rightarrow \Tau}{\arginf} \ \poa (\gee,\sdist,T),
\end{equation*}
such that it is the scaled marginal-cost taxation mechanism which minimizes the price of anarchy expressed in \eqref{eq_poa_def} for a given family of networks~$\gee$ and sensitivity distributions~$\sdist$. Further, we define the price of anarchy bound under an optimal tolling mechanism as
\begin{equation}
\poa^*(\gee,\sdist) \triangleq \underset{T: L(\gee) \rightarrow \Tau}{\rm inf} \ \poa (\gee,\sdist,T).
\end{equation}
In this paper, we consider networks with two parallel links (the family of which we simply denote $\gee$), to demonstrate the value of information to a system designer by comparing the price of anarchy bounds when different amounts of information are available. It is not currently known if these bounds extend to richer classes of networks, but these networks often display worst-case inefficiency over larger classes of networks and allow us to analyze the benefit of these partially informed tolls~\cite{Roughgarden2003}.

The information available to the system designer is encoded in these families of networks and sensitivity distributions. When the system designer is \emph{network-aware}, the family of networks for which the taxation mechanism is defined is a singleton (i.e., the system designer knows the network). Otherwise, we say the system designer is \emph{network-agnostic} when only a family of possible networks is known. We look at the following four cases when the system designer is aware of the network, mean sensitivity, both, or neither:

\noindent \textbf{\textit{Network-agnostic \& mean-agnostic:}} We first consider the case that the system designer knows only the lower and upper bound on user's sensitivities. This was originally considered in%
~\cite{Brown2017d} for parallel routing problems under a utilization constraint; here, we generalize this known result for networks with no utilization constraint. The price of anarchy bound under an optimal toll in this scenario is denoted
\begin{equation}
 \poa^*(\gee,\sdist) = \underset{T: L(\gee) \rightarrow \Tau}{\rm inf} \ \poa (\gee,\sdist,T),
\end{equation}
and a general expression is given in Theorem~\ref{thm_net_agn}. A realization of this is given in box (A) of Figure~\ref{fig_table}.

\noindent \textbf{\textit{Network-agnostic \& mean-aware:}} Second, we consider the scenario in which the system designer knows the mean of users' sensitivities, as well as the lower and upper bound. Here, the number of possible sensitivity distributions is reduced, and the system designer will be able to make better guarantees on worst-case inefficiency. The price of anarchy bound under an optimal toll in this scenario is denoted
\begin{equation}
 \poa^*(\gee,\sdistm) = \underset{T: L(\gee) \rightarrow \Tau}{\rm inf} \ \poa (\gee,\sdistm,T),
\end{equation}
and a general expression is given in Theorem~\ref{thm_avg}. A realization of this is given in box (B) of Figure~\ref{fig_table}.

\noindent \textbf{\textit{Network-aware \& mean-agnostic:}} Next, we consider when the system designer again knows only the lower and upper bound on users' sensitivities, but is also aware of the network structure. The system designer will be able to make better guarantees on worst-case inefficiency when there is no uncertainty in the network that was realized. The price of anarchy under an optimal toll in this scenario is denoted
\begin{equation}
 \poa^*(G,\sdist) = \underset{T: L(G) \rightarrow \Tau}{\rm inf} \ \poa (G,\sdist,T),
\end{equation}
and a general expression is given in Theorem~\ref{thm_net}. A realization of this is given in box (C) of Figure~\ref{fig_table}.

\noindent \textbf{\textit{Network-aware \& mean-aware:}} Finally, we consider the case when the system designer knows the lower bound, upper bound, and mean of users' sensitivities and also knows the network structure. Here, the system designer will be able to guarantee the best performance of any of the scenarios we consider. The price of anarchy bound under an optimal toll in this scenario is denoted
\begin{equation}
 \poa^*(G,\sdistm) = \underset{T: L(G) \rightarrow \Tau}{\rm inf} \ \poa (G,\sdistm,T),
\end{equation}
and a general expression is given in Theorem~\ref{thm_all}. A realization of this is given in box (D) of Figure~\ref{fig_table}.

Though it is clear that a network-aware mechanism will provide better guarantees than a similar network-agnostic mechanism, and that a mean-aware mechanism will provide better guarantees than a similar mean-agnostic mechanism, it is not obvious what will provide a better gain in performance when introduced to an uninformed system designer: network-awareness or mean-awareness.

\begin{figure}
\centering
\includegraphics[width=0.45\textwidth]{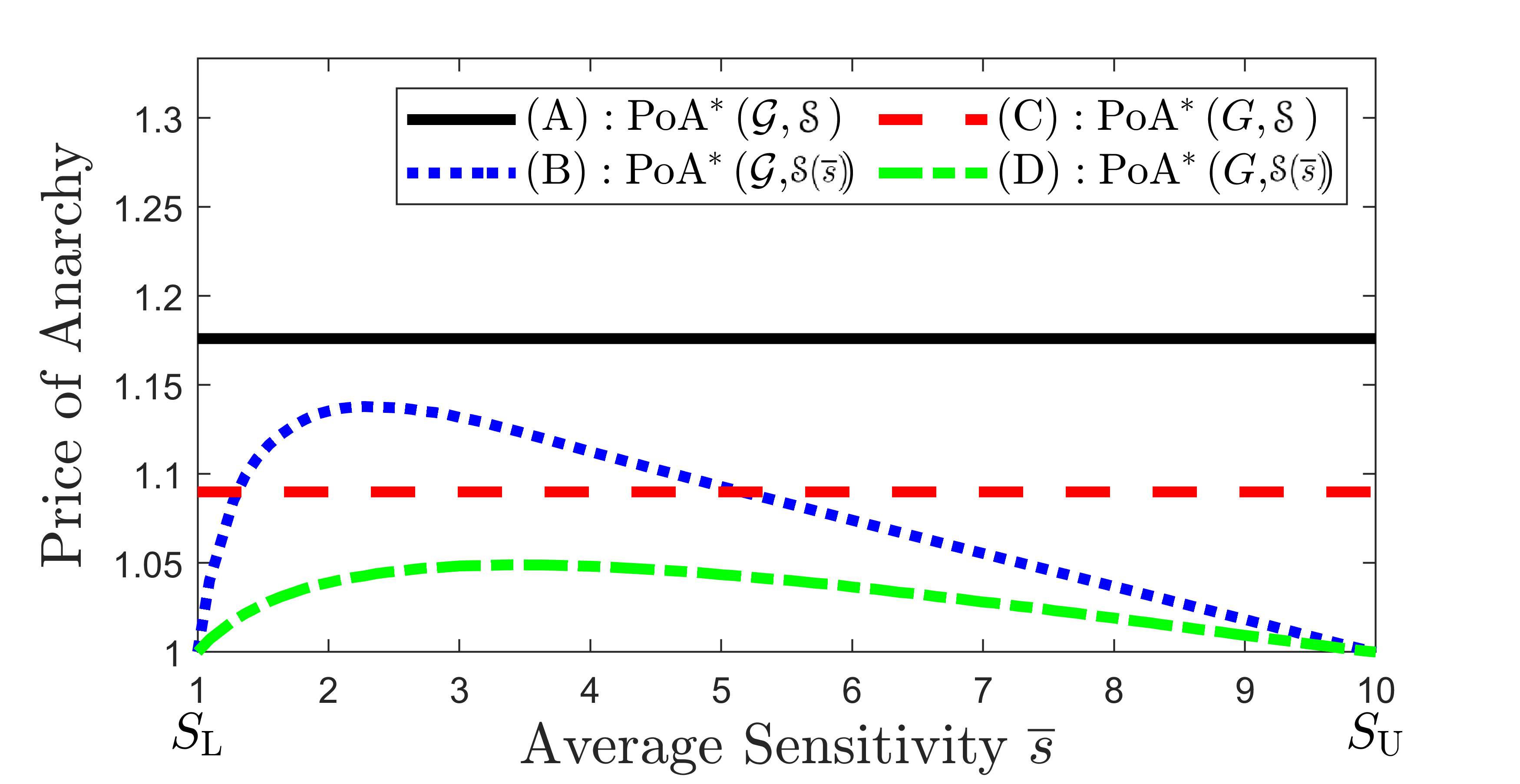}
\caption{Price of anarchy with respect to the mean sensitivity. Each plot represents a bound for one of the four introduced tolling mechanisms: (A) network agnostic, mean agnostic toll, (B) network agnostic, mean aware toll, (C) network aware, mean agnostic toll, and (D) network aware, mean aware toll. Each is below the untolled price of anarchy of 4/3. Price sensitivity bounds $\bmin=1$ and $\bmax=10$ are shown; changing these values has a minimal effect on the relation between the lines.}
\label{fig_PoA}
\end{figure}

To bring clarity to the results of this paper, we plot the price of anarchy bounds with respect to~$\sbar$ for the case $\bmin=1$ and $\bmax=10$ in Figure~\ref{fig_PoA}. As expected, the addition of information provides better guarantees on worst-case inefficiency. Notice when $\overline{s}=\bmin$ or $\overline{s}=\bmax$, the price of anarchy is 1 for the mean-aware tolls, i.e., the toll is optimal for any network. When the mean of user sensitivity is at one of the bounds, then the only feasible distribution is each user sharing that sensitivity. When all users share a common sensitivity, tolls can by applied to induce optimal routing~\cite{Beckmann1956}. However, the network-aware, mean-agnostic taxation mechanism performs better than the network-agnostic, mean-aware for means that are far from~$\bmin$ and~$\bmax$. This implies that knowing the mean sensitivity of users is not enough to guarantee better performance than knowing the structure of the network.

\section{Main Results}\label{section_main_results}
For each scenario we consider, we give a theorem reporting the price of anarchy bound, as well as a proposition reporting the scaling factor of the optimal scaled marginal-cost toll.
Note that in this paper, we focus our search for optimal tolls to a search for scaled marginal-cost tolls. These taxation mechanisms can be defined by a single scaling factor $k$, and will be denoted $\Tsmc(k)$. It is shown in Lemma~2.2 in~\cite{Brown2017d} that the search for an optimal \emph{bounded} toll is equivalent to a search for an optimal linear toll in the network-agnostic case. The form of the optimal taxation mechanism in the network aware case remains an open question.

\subsection{Network-agnostic, mean-agnostic}
We first look at the case where the system designer has the least amount of information available to them, and formulate a price of anarchy bound under a taxation mechanism that is optimal in this scenario.
\begin{theorem} \label{thm_net_agn}
When only $\bmin$ and $\bmax$ are known, the price of anarchy under an optimal, scaled marginal-cost tolling mechanism is
\begin{equation} \label{eq_poa-nomean-agn}
\poa^* \left( \gee, \sdist \right) = \frac{\left(q-1+\sqrt{q^2+14q+1}\right)^2}{8q\left(-q-1+\sqrt{q^2+14q+1}\right)},
\end{equation}
where $q:=\bmin/\bmax$.
\end{theorem}

To find this upper bound, we find the optimal scaled marginal-cost toll. 
We then use this toll to find the worst-case game which maximizes the price of anarchy under this toll.

\begin{proposition}\label{prop_k_agn}
When only $\bmin$ and $\bmax$ are known, the optimal network-agnostic marginal-cost toll scaling factor is
\begin{equation} \label{eq_kagn}
\kagn = \frac{-\bmin -\bmax+\sqrt{\bmin^2+14\bmin\bmax + \bmax^2}}{2\bmin\bmax}. 
\end{equation}
\end{proposition}\vspace{1mm}
\begin{proof}
The expression~\eqref{eq_kagn} can be shown to be a solution to the equation
\begin{equation} \label{eq:equalpoa}
\frac{4}{4\left(1+\kagn\bmin\right) - \left(1+\kagn\bmin\right)^2} = \frac{\left(1+\kagn\bmax\right)^2}{4\kagn\bmax}.
\end{equation}
It is shown in~\cite{Brown2017e} that when $\bmin<\bmax$,~\eqref{eq:equalpoa} always has a exactly one solution on the interval $\left(1/\bmax,1/\bmin\right)$, and that solution is the desired optimal scale factor.
We show here that~\eqref{eq_kagn} describes this particular solution.
Let $p=\bmax/\bmin$. 
In this regime, $p>1$, so we have
\begin{align}
1+14p + p^2 	> 1+14p + p^2 + 8(1-p) \nonumber 
							= \left(p+3\right)^2.
\end{align}
Thus,
\begin{align}
\kagn 	&> \frac{-1 -p+\sqrt{\left(p+3\right)^2}}{2\bmax} = \frac{1}{\bmax}.
\end{align}

Likewise, letting $q=\bmin/\bmax$ (so that $q<1$), we have
\begin{equation*}
1+14q+q^2 < 1+14q+q^2 + 8(1-q) = (q+3)^2,
\end{equation*}
yielding
\begin{equation}
\kagn 	< \frac{-1 -q+\sqrt{\left(q+3\right)^2}}{2\bmax} = \frac{1}{\bmin}.
\end{equation}
\end{proof}

\noindent \emph{Proof of Theorem~\ref{thm_net_agn}:} Using the optimal scaling factor from Proposition~\ref{prop_k_agn}, expression~\eqref{eq_poa-nomean-agn} can be found by substituting~\eqref{eq_kagn} into~\eqref{eq:equalpoa}. $\qed$

\subsection{Network-agnostic, mean-aware}
When the mean sensitivity of users in the population is available to the system designer, the set of possible distributions is reduced to the set $\sdistm \subset \sdist$. This additional information allows the system designer to find a new optimal taxation mechanism that improves the inefficiency bound.

\begin{theorem}\label{thm_avg}
When $\bmin$, $\bmax$ and the mean sensitivity $\sbar$ are known, the price of anarchy under an optimal, scaled marginal-cost toll is given by
\begin{equation}
 \poa^* \left(\gee, \sdistm \right) = \frac{R^2-\alpha R+\alpha}{\alpha-\alpha^2/4},
\end{equation}
where $R:=(\bmax-\sbar)/(\bmax-\bmin)$, and $\alpha = (1+\bmax \ks)R$, with $\ks$ being the solution to~\eqref{eq_k_avg}. \vspace{.5mm}
\end{theorem}

To find this bound, we perform a series of reductions in the set of distributions and networks we need consider in our search for one which realizes worst-case inefficiency. First, we show that for each network there exist two sensitivity distributions which realize the price of anarchy over all possible distributions. Then, we demonstrate a method to transform a network to one that has one linear and one constant latency function and higher price of anarchy. Finally, we prove the optimal scaling factor for a scaled marginal-cost toll equates the price of anarchy for two specific networks and show that these realize the given price of anarchy bound.

We first give a list of supporting Lemmas, then use these to complete the proof of the theorem.


We say users $x,y$ have the same \emph{type} if $s(x)=s(y)$. Further let a \emph{bimodal} distribution be one in which there exist exactly two user types; the set of such distributions is denoted  $\sdistmbi \subset \sdistm$. We denote a bimodal distribution with types $S_1$ and $S_2$ by $(S_1,S_2)$. Note that for a given $\sbar$, $S_1$ and $S_2$, the mass of users with each sensitivity is well defined. Additionally, we adopt the convention used elsewhere that the network links are indexed such that $b_1\leq b_2$.
\begin{lemma}\label{lemma_bimodal}
A Nash flow $f$ for a sensitivity distribution $s \in \sdistm$, under a scaled marginal cost tax $T$, is likewise a Nash flow for some distribution $s' \in \sdistmbi$ in which one type of user is indifferent between the two edges and all users on each edge are of a single type. This implies the price of anarchy over sensitivities in $\sdistm$ is equal to the price of anarchy over bimodal distributions in $\sdistmbi$, i.e.,\vspace{-1mm}
\begin{equation}
    \poa(\gee,\sdistm,T) = \poa(\gee,\sdistmbi,T).\vspace{2mm}
\end{equation}
\end{lemma}
\begin{proof}
Let $s_1 \in \sdistm$ be some distribution of users' sensitivities, and let $S_{\rm ind}$ be the sensitivity that has equal cost between the two links in the Nash flow $\fnash$, i.e., solution to
\begin{equation} \label{eq_sind_def}
(1+S_{\rm ind}k)a_1\fnash_1 + b_1 = (1+S_{\rm ind}k) a_2\fnash_2+b_2.
\end{equation}
Note that in the case where $S_{\rm ind}>\bmax$ or $S_{\rm ind}<\bmin$, any distribution $s \in \sdist$ will have the same Nash flow with all users choosing the same edge. First, consider the case where $S_{\rm ind} < \mu(s_1)$, where $\mu(\cdot)$ is the mean of the distribution. From Claim 1.1.2 in~\cite{Brown2017c}, if a user has a sensitivity $S<S_{\rm ind}$, then they strictly prefer the first link; if they have a sensitivity $S>S_{\rm ind}$ then they strictly prefer the second.

Now, let $s_2$ be a new distribution where each user who had chosen edge 1 now has sensitivity $S_{\rm ind}$. The Nash flows from $s_1$ and $s_2$ are the same, as the same number of users have a sensitivity $S\leq S_{\rm ind}$ and thus the same users choose the first edge. It is clear that $ \mu(s_2)>\mu(s_1)$ as no user has a lower sensitivity and some have higher.

Now, consider a third distribution $s_3$, where users who chose edge 2 now have some sensitivity  $S' \in (S_{\rm ind},\bmax]$; these users will now strictly prefer the second edge of the network but the Nash flow will remain unchanged. If we pick $S'=\bmax$, the mean has surely increased again; if we pick $S' = S_{\rm ind}$, because we are in the case $S_{\rm ind} < \sbar$, the mean is lower than $\mu(s_1)$. Because $\mu(s_3)$ is continuous with $S'$, we can select $S'$ so that $\mu(s_3)=\mu(s_1)$. The case of $S_{\rm ind} > \mu(s_1)$ is similar.

The distribution $s_3 = (S_{\rm ind}, S')$ induces the same Nash flow as $s_1$ and now, one set of users is indifferent and users of the same type exist on the same edge only.
\end{proof}

Having shown in Lemma~\ref{lemma_bimodal} that the price of anarchy is realized by bimodal distributions, we further refine our search for worst-case populations to just two specific bimodal distributions with simple characterizations.

\begin{lemma}\label{lemma_slow_supp}
For a given network $G \in \gee$ and scaled marginal-cost tax $T$ with toll scaling factor $k$, two distributions $\slow$ and $\supp$, that maximize and minimize (respectively) the flow on the first edge of the network, realize the price of anarchy over those in $\sdistmbi$,
\begin{equation}
    \poa(G,\sdistmbi,T) = \poa(G,\{ \slow,\supp \},T)
\end{equation}
\end{lemma}
\begin{proof}
The proof follows from the fact that total latency is quadratic with the flow, thus the largest price of anarchy will come from the flow that is furthest from optimal. From Lemma~\ref{lemma_bimodal}, we see that any flow induced by a distribution $s \in \sdistm$ can be realized by a bimodal distribution that has one set of users observing equal cost between the links and each edge containing only one sensitivity type. We therefore define $\slow$ as the distribution that maximizes $\fnash_1$ and $\supp$ as the distribution which maximizes $\fnash_1$.
\end{proof}

We next reduce our search for a worst-case network to a set of graphs that have a linear latency function on the first edge and a constant latency on the second, denoted $\geelc$. 
\begin{lemma} \label{lemma_glc}
For any $G \in \gee$, there exists a $\hat{G} \in \geelc \subset \gee$ that, under the same scaled marginal cost tolling mechanism $\Tsmc(k)$, has a higher price of anarchy, implying,
\begin{equation}
    \text{PoA}(\gee,\sdistm,\Tsmc(k)) = \text{PoA}(\geelc,\sdistm,\Tsmc(k)).
\end{equation}
\end{lemma}

\begin{proof}
Consider a network $G \in \gee$ with affine latency functions on each link $\ell_i(f) = a_if+b_i$. Let $\hat{G}$ have cost functions $\hat{\ell}_i(f) = \hat{a}_if+\hat{b}_i$ with $\hat{a}_i\geq0$ and $\hat{b}_i\geq0$. 

We first show that simply removing the constant latency term on the first edge $b_1$ strictly increases the price of anarchy under any scaled marginal-cost toll.

Using the optimal and Nash flow in \eqref{eq_Nash_and_opt_flow}, if $\hat{b}_2 = b_2-b_1$ and $\hat{b}_1=0$ then $G$ and $\hat{G}$ will have the same optimal flow and Nash flow for a distribution $s$. From \eqref{eq_total_lat}, we observe that $\Lopt(G) = \Lopt(\hat{G}) +b_1$ as well as $\Lnash(G,s,k) = \Lnash(\hat{G},s,k) +b_1$; therefore,
\begin{multline*}
\text{PoA}\left(G,\sdistm, \Tsmc(k)\right)  =  \frac{\Lnash \left(\hat{G},s,\Tsmc(k)\right)+b_1}{\Lopt(\hat{G})+b_1} \\  \leq \frac{\Lnash \left(\hat{G},s,\Tsmc(k)\right)}{\Lopt(\hat{G})} = \text{PoA}(\hat{G},\sdistm,\Tsmc(k)).
\end{multline*}
Thus, for any network $G \in \gee$, there exists a network $\hat{G}$ with a linear latency function on an edge with higher price of anarchy.

Next, we show a network $G \in \gee$ will have the same price of anarchy as a network $\hat{G} \in \gee$ under the same scaled marginal cost toll if the latency functions of $\hat{G}$ equal the latency functions of $G$ times a scaling factor $c$.

Under a distribution $s \in \sdist$, $G$ and $\hat{G}$ will have the same Nash flow. Using the indifferent sensitivity $S_{\rm ind}$ that is the solution to \eqref{eq_sind_def}, the Nash flow and optimal flow on the first edge are
\begin{equation}\label{eq_Nash_and_opt_flow}
f_1^{\rm opt} = \frac{2a_2+b_2-b_1}{2(a_1+a_2)}, \quad \fnash_1 = \frac{(1+S_{\rm ind}k)a_2+b_2-b_1}{(1+S_{\rm ind}k)(a_1+a_2)}.
\end{equation}
Under the same distribution $s$, the same sensitivity $S_{\rm ind}$ will satisfy
\begin{equation}
	(1+S_{\rm ind} k) c a_1 \fnash_1 +c b_1 = (1+S_{\rm ind} k)c a_2 \fnash_2 +c b_2,
\end{equation}
which are the latency functions for the network $\hat{G}$. It is now clear that $G$ and $\hat{G}$ will have the same Nash and optimal flows. From the definition of total latency in \eqref{eq_total_lat}, the latency in $\hat{G}$ will be $c$ times the latency in $G$ under the same flow. The price of anarchy, which is the ratio of two  total latencies, will be identical in $G$ and $G'$.

Lastly, we show that by decreasing $a_2$ in a network, the price of anarchy will increase. In Lemma~\ref{lemma_bimodal}, it was shown that any feasible Nash flow can be induced by a bimodal sensitivity distribution in which users are segregated on either link by their sensitivity. The price of anarchy for the network $G$ with a Nash flow caused by $s$ will therefore be,
\begin{equation} \label{PoA_bimodal}
    \text{PoA}(G,s,\Tsmc(k)) = \frac{\ell_1(\fnash_1)\fnash_1+\ell_2(\fnash_2)\fnash_2}{\ell_1(f_1^{\rm opt})f_1^{\rm opt}+\ell_2(f_2^{\rm opt})f_2^{\rm opt}}.
\end{equation}
Let us consider the case where $\fnash_2> \fopt_2$. Now, consider a new network, $\hat{G}$ which replaces latency function $\ell_2(f) = a_2f+b_2$ in $G$ with $\hat{\ell}_2(f)= a_2f+\hat{b}_2$ where $\hat{b}_2 = b_2 + \delta$ such that $\delta > 0$. Because the users are segregated on the links, the Nash flow will not change. Note that because $\fnash_2> \fopt_2$
\begin{equation*}
    \frac{\ell_2(\fnash_2)}{\ell_2(f_2^{\rm opt})} = \frac{a_2\fnash_2+b_2}{a_2f_2^{\rm opt}+b_2} < \frac{\fnash_2}{f_2^{\rm opt}}.
\end{equation*}
It can now be shown that
\begin{align*}
\frac{\Lnash(G,s,\Tsmc(k))}{\Lopt(G)} &< \frac{\Lnash(G,s,\Tsmc(k))+\delta \fnash_2}{\Lopt(G)+\delta \fopt_2} \\ &= \frac{\Lnash(\hat{G},s,\Tsmc(k))}{\Lopt(\hat{G})}.
\end{align*}
Thus the price of anarchy has increased in the new network $\hat{G}$, under the same sensitivity distribution and scaled marginal cost toll, when $b_2$ was increased, which has the same effect as decreasing the other terms and holding $b_2$ constant. A very similar argument can be followed for when $\fnash_2 < f_2^{\rm opt}$ by picking $\hat{a_2}=a_2 - \delta$, and the price of anarchy again increases.
\end{proof}

We further reduce our search for a worst case network to two specific networks. For a given set of distributions $\sdistm$ and toll scaling factor $k$, we define two networks: 

(1)~$G_\beta~\in~\geelc$ with latency functions $\ell_1(f_1)=f_1$ and $\ell_2(f_2)=\beta$ and satisfies $s_\mathpzc{l}^{(\sbar,G_\beta,k)} = (\bmin,\bmax)$, and 

(2)~$G_\alpha~\in~\geelc$  with latency functions $\ell_1(f_1)=f_1$ and $\ell_2(f_2)=\alpha$ and satisfies $s_\mathpzc{u}^{(\sbar,G_\alpha,k)} = (\bmin,\bmax)$. Due to the discussion in the proof of Lemma~\ref{lemma_glc}, any network in $\geelc$ with cost functions satisfying $b_2/a_1 = \beta$ will have the same price of anarchy as $G_\beta$, and the same is true for for $G_\alpha$.

\begin{lemma}\label{lemma_Ga_Gb}
For linear constant networks, under sensitivity distributions in $\sdistm$ with toll scaling factor $k$, the network $G_\alpha$ or $G_\beta$ will realize the upper bound on the price of anarchy, i.e.,
\begin{align*}
\poa(\geelc,\sdistm,\Tsmc(k))
 = \poa(\{ G_\alpha, G_\beta \},\sdistm,\Tsmc(k)).
\end{align*}
\end{lemma}
\begin{proof}
It can be seen by differentiation of \eqref{eq_Nash_and_opt_flow}, the price of anarchy increases with the value of the indifferent sensitivity when $\fnash_1 < f_1^{\rm opt}$ and decreases when $\fnash_1 < f_1^{\rm opt}$. Recall that $\slow$ has $f_{1\mathpzc{l}} > f_1^{\rm opt}$ and indifferent sensitivity $S_{\mathpzc{l}1}$; similarly, $\supp$ has $f_{1\mathpzc{u}} < f_1^{\rm opt}$ and indifferent sensitivity $S_{\mathpzc{u}2}$. It is therefore true that having $S_{\mathpzc{l}1}=\bmin$ or $S_{\mathpzc{u}2}=\bmax$ is a necessary condition for the network which maximizes the price of anarchy. 

Further, in bimodal distributions $(S_1,S_2)$ where users are homogeneous on either link, $\fnash_1 = (S_2-\sbar)/(S_2-S_1)$. For $\slow$ when users with sensitivity $S_{\mathpzc{l}1}=\bmin$ are indifferent, the largest flow that can occur on $f_1$ occurs when $\slow = (\bmin,\bmax)$. Similarly, for $\supp$, when users with sensitivity $S_{\mathpzc{u}2}=\bmax$ are indifferent, the least flow in $f_1$ has $\supp = (\bmin,\bmax)$. One of these two conditions must be met by a network $G\in \geelc$ to maximize the price of anarchy. Those networks are the defined $G_\alpha$ and $G_\beta$.
\end{proof}

\begin{proposition}\label{prop_k_avg}
When $\bmin$, $\bmax$ and the mean sensitivity $\sbar$ are known, the optimal network-agnostic marginal-cost toll scaling factor $\ks$ will be the solution on $(1/\bmax,1/\bmin)$ to
\begin{align}\label{eq_k_avg}
&\frac{4(1+\bmax k-   \bmax k R)}{(4+R)(1+\bmax k)+(\bmax k + \bmax^2 k^2)R} \nonumber \\ 
&\qquad= \frac{4(1+\bmin k- \bmin k R)}{(4+R)(1+\bmin k)+(\bmin k + \bmin^2 k^2)R},
\end{align}
\end{proposition}
\vspace{0.1in}
\begin{proof}
The $k$ that solves \eqref{eq_k_avg} equates the price of anarchy for $G_\beta$ and $G_\alpha$. To show this is optimal, it is sufficient to show that the price of anarchy for the network $G_\beta$ is monotone decreasing with $k$ while the price of anarchy for the network $G_\alpha$ is monotone increasing with $k$.
If the networks have this monotonic relation with $k$, then the $k$ that minimizes the price of anarchy  must equalize them.

Consider a network $G \in \geelc$ characterized by $\gamma = b_2/a_1 $. If this satisfies that $\slow =(\bmin,\bmax)$ or $\supp = (\bmin,\bmax)$, then the price of anarchy for this network will be
\begin{equation}\label{eq_network_ratio}
 \poa(G,\sdistm,\Tsmc) = \frac{R^2-\gamma R+\gamma}{\gamma-\gamma^2/4}.
\end{equation}
This expression is locally minimized by $\gamma = 2R$, further, by differentiation, it can be observed that it is monotone decreasing for $0 < \gamma < 2R$ and monotone increasing for $\gamma > 2R$.

For the previously defined network $G_\beta$, under the bimodal distribution $(\bmin,\bmax)$,
\begin{equation}\label{eq_beta_equal}
	(1+\bmin k)R = \beta,
\end{equation}
where $\beta$ is dependent on the scaling factor $k$. From \cite{Brown2017e}, the optimal scaling factor $k$ will be in $(1/\bmax,1/\bmin)$. Therefore, for any $k$, $\beta < 2R$. The price of anarchy for this network is therefore monotone decreasing with $\beta$, and from \eqref{eq_beta_equal}, $\beta$ is clearly increasing with $k$. The price of anarchy of the network is therefore decreasing with $k$. Similarly for $G_\alpha$, under the distribution $(\bmin,\bmax)$,
\begin{equation}\label{eq_alpha_equals}
	(1+\bmax k)R = \alpha > 2R,
\end{equation}
and the price of anarchy will be increasing with $k$.
\end{proof}

\noindent \emph{Proof of Theorem~\ref{thm_avg}:}
From Lemma~\ref{lemma_Ga_Gb}, a network $G_\alpha$ realizes the price of anarchy when the toll scaling factor is chosen optimally as in Proposition~\ref{prop_k_avg}. The price of anarchy for this network is found by substituting $\alpha$ from \eqref{eq_alpha_equals} into the latency function ratio in \eqref{eq_network_ratio}.
$\qed$

\subsection{Network-aware, mean-agnostic}
If the system designer is made aware of the structure of the network, then the optimal network-aware taxation mechanism will out-perform one that is designed without such knowledge. Though it is still an open question as to what tolls are optimal for a network-aware system designer in this scenario, we opt to still consider linear tolling functions to provide a means of comparing with the previously introduced, network-agnostic taxation mechanisms.
\begin{theorem}\label{thm_net}
When only $\bmin$ and $\bmax$ are known, the price of anarchy under an optimal network-aware scaled marginal-cost toll is tightly upper bounded by
\begin{equation} \label{eq_poa-nomean-awa}
 \poa^* \left(G,\sdist \right) \leq \frac{4}{3}\left(1-\frac{\sqrt{q}}{\left(1+\sqrt{q}\right)^2}\right)
\end{equation}
where $q:=\bmin/\bmax$.
\end{theorem}

For the proof of this bound, we implicitly assume that the system operator can determine hypothetical homogeneous low-sensitivity Nash flows associated with each tolling factor; this is reasonable since the Nash flow for a homogeneous population is known to be the solution to a convex optimization problem~\cite{Roughgarden2002}.
In the following, we write $\fnash_i(G,S,k)$ to denote the amount of traffic on link $i$ in a Nash flow on network $G$ with homogeneous sensitivity $S$ and toll scale factor $k$.

\begin{proposition}
For any network~$G\in\gee$ and any~$\bmax\geq\bmin>0$, let $k^{\rm gm}=\left(\bmin\bmax\right)^{-1/2}$. 
The following is an optimal network-aware marginal-cost toll scaling factor:
\begin{equation}
    \kg = \left\{
    \begin{array}{cl}
    0 & \mbox{ if } \fnash_2(G,\bmin,k^{\rm gm})=0, \\
    k^{\rm gm} & \mbox{ otherwise.}
    \end{array}\right.
\end{equation}\vspace{.1mm}
\end{proposition}
\begin{proof}
Consider the following cases, differentiated by the structure of Nash flows resulting from $k=k^{\rm gm}:=(\bmin\bmax)^{-1/2}$:
\begin{enumerate}
\item $\fnash_2(G,\bmin,k^{\rm gm})>0$ , and
\item $\fnash_2(G,\bmin,k^{\rm gm})=0$.
\end{enumerate}

It is shown in~\cite{Brown2017c} that in Case (1), it must be true that $\Lnash\left(G,\bmin,k\right) = \Lnash\left(G,\bmax,k\right)$ and that this choice of $k$ is uniquely optimal, resulting in the price of anarchy given in~\eqref{eq_poa-nomean-awa}.

Consider Case (2).
Here, the extreme low-sensitivity population with $s=\bmin$ strictly prefers link $1$ when $k=k^{\rm gm}$, effectively stripping the designer of her influence over the price of anarchy.
It can easily be shown (using, e.g., tools from~\cite{Brown2017c}) that 
\begin{equation}
k\rq{}\leq k^{\rm gm}\ \implies \ \Lnash\left(G,\bmin,k\right)=\Lnash\left(G,\emptyset\right), \label{eq:cantchangebmin}
\end{equation}
but that 
\begin{equation}
k^\dagger>k^{\rm gm} \ \implies \ \Lnash\left(G,\bmax,k^\dagger\right)>\Lnash\left(G,\emptyset\right).
\end{equation}
That is, in this regime, the designer cannot change the behavior of $s=\bmin$ without increasing tolls, but cannot increase tolls because this would cause the high-sensitivity population with $s=\bmax$ to route more inefficiently.
That is, $k=0$ is an optimal tolling coefficient in this case.%
\footnote{%
In this case, the set of price-of-anarchy-minimizing tolling coefficients is not a singleton in general: any coefficient satisfying $\Lnash(G,\bmin,k)\geq\Lnash(G,\bmax,k)$ is optimal.
Implication~\eqref{eq:cantchangebmin} means that this set always contains $k=0$.}
\end{proof}

\noindent \emph{Proof of Theorem~\ref{thm_net}} It follows easily from the results in~\cite{Brown2017c} that in Case (2) when $k\leq k^{\rm gm}$, it is true for any $s$ that $\Lnash(G,s,k)\leq\Lnash(G,\bmax,k^{\rm gm})$; the price of anarchy bound for this scenario is thus precisely that in~\cite{Brown2017c}, where now we include games which need not have flow on every edge in an untolled Nash flow. $\qed$

\subsection{Network-aware, mean-aware}
Finally, we consider the network-aware, mean-aware scaled marginal-cost toll to illustrate the gain in performance when both network and population information is available.
\begin{theorem} \label{thm_all}
When $\bmin$, $\bmax$ and the mean sensitivity $\sbar$ are known, under an optimal network-aware scaled marginal-cost tolling mechanism, the price of anarchy is tightly upper bounded by
\begin{equation}
	\poa^* \left(G,\sdistm \right) \leq \frac{R^2-\beta R+\beta}{\beta-\beta^2/4},
\end{equation}
where $R = (\bmax-\sbar)/(\bmax-\bmin)$ and $\beta$ is the unique solution on the interval $[0,2]$ to
\begin{equation}\label{eq_bawa}
    \beta = R\left( 1+\sqrt{\frac{1+R-\beta}{\sbar/\bmin+R-\beta}} \right).
\end{equation} \vspace{1mm}
\end{theorem}

We prove the theorem by making several of the same reductions as in Theorem~\ref{thm_avg}. However, we now derive an optimal network-aware toll scaling factor $\ksg$ and show a unique network will realize the upper bound on inefficiency in this setting.

\begin{proposition}\label{prop_opt_awa}
For a network $G\in \gee$ with price sensitivity distributions $s \in \sdistm$  with extreme sensitivity distributions $\slow = (S_{\mathpzc{l}1},S_{\mathpzc{l}2})$ and $\supp = (S_{\mathpzc{u}1},S_{\mathpzc{u}2})$, the optimal toll scaling factor for a scaled marginal cost toll will take the form,
\begin{equation}
    k_{(\sbar,G)} = \frac{1}{\sqrt{\slone \sutwo}}.
\end{equation}
\end{proposition}
\emph{Proof of Proposition~\ref{prop_opt_awa}:}
From Lemma~\ref{lemma_bimodal}, under the same tolling mechanism, the set of Nash flows caused by $\sdistm$ is equivalent to those caused by distributions with bounds $[\slone,\sutwo]$ and no mean constraint. The optimal scaling factor will therefore minimize the price of anarchy over this set of distributions. From \cite{Brown2017c}, the optimal scaling factor for a scaled marginal-cost toll will take this form.
$\qed$

In Lemma~\ref{lemma_glc}, it was shown that a transformation from a network $G\in \gee$ to a network $\hat{G} \in \geelc$ will increase the price of anarchy; we also note that this transformation had no dependence on the toll scaling factor and we can thus choose a $k$ that is optimal for the resulting network.
\begin{corollary}\label{cor_Tsg}
When making a reduction from $G\in \gee$ to $\hat{G}\in \geelc$, the price of anarchy increases regardless of the toll scaling factor $k$, including when $k = \ksg$ for each network before and after the reduction. 
\end{corollary}
\emph{Proof:} In the proof of Lemma \ref{lemma_glc}, the relation between $k$ and the price of anarchy was not used; instead, it was shown that the price of anarchy increases as the network is transformed from any two link network, to one that was in $\geelc$.
Consider having network $G$ with the non-optimal toll scaling factor $k_{(\sbar,\hat{G})}$. When the reduction from $G$ to $\hat{G}$ is done, by Lemma \ref{lemma_glc} we have
\begin{align*}
    \text{PoA}(G,\sdistm,\Tsmc(k_{(\sbar,G)}) & \leq \text{PoA}(G,\sdistm,\Tsmc(k_{(\sbar,\hat{G})})\\
    &\leq \text{PoA}(\hat{G},\sdistm,\Tsmc(k_{(\sbar,\hat{G})}). ~~~~~~~ \qed
\end{align*}

\noindent \emph{Proof of Theorem~\ref{thm_all}:}
It is shown in Lemma~\ref{lemma_Ga_Gb} that a set of two networks realizes the price of anarchy. The price of anarchy for the network $G_\beta$ is found by \eqref{eq_network_ratio}. Now, let $G'$, defined by $\beta'$, be a network that has the same price of anarchy when the flow $R$ is on the first link, 
One solution  is clearly $\beta=\beta'$, the other is $\beta' = \frac{(4-\beta)R}{R^2-\beta R +\beta}$. Using the cost function of network $G$, we have $\beta = (1+\bmin k_{(\sbar,G)})R$. Thus, if $\beta'$ satisfied $\supp = (\bmin,\bmax)$ for the same mean sensitivity, then
\begin{align}\label{eq_betas_equal_a}
    \beta' &= \frac{(4-\beta)R}{R^2-\beta R +\beta} = (1+\bmax k_{(\sbar,G')})R.
\end{align}
However, it can be shown that the right hand side of \eqref{eq_beta_equal} is strictly less than the left hand side. This imposes that the flow $f_1=R$ can not be a Nash flow in $G'$ under distributions in $\sdistm$ and therefore not achieve the same price of anarchy as $G$. This implies that the price of anarchy for $G_\beta$ is greater than that of $G_\alpha$ when both are tolled optimally with respect to Proposition~\ref{prop_opt_awa}.

As this network is optimally tolled, from Proposition~\ref{prop_opt_awa}, it will be the case that
\begin{eqnarray}
    S_{\mathpzc{u}2} = \frac{\sbar-\bmin}{1+R-\beta} +\bmin. \label{s_2_p_2}
\end{eqnarray}
Now, in sensitivity distribution $\slow = (\bmin,\bmax)$, users with sensitivity $\bmin$ are indifferent with optimally scaled toll $\ksg$. Using $\beta$ from \eqref{eq_beta_equal} and substituting the optimal scaling factor with extreme sensitivity from \eqref{s_2_p_2} leads to the characterization of $\beta$ in the theorem statement, and the price of anarchy is found by substituting this into \eqref{eq_network_ratio}.
$\qed$

\section{Conclusion}
This paper represents a study of the challenges and opportunities afforded to a system designer when faced with information.
While it is clear that gathering additional information about a problem can help a designer to influence behavior more effectively, the question of how to use the information optimally is not trivial, even for simple classes of problems.
Ongoing work focuses on extending this paper's results to more general classes of networks and determining when precisely scaled marginal-cost taxes are optimal over all taxation mechanisms.


\bibliographystyle{IEEEtran}
\bibliography{library}

\end{document}